\documentclass{llncs}
\usepackage{amssymb,amsmath,enumerate,latexsym}
\usepackage{color,algorithmic,xspace,theorem}

\usepackage[margin=3cm,paper=letterpaper]{geometry}
\theoremstyle{remark}

\newcounter{tbsnr}\newenvironment{tbs}
{\addtocounter{tbsnr}{1}\par\bigskip \noindent\fbox{\thetbsnr}
\hspace{2mm}\begin{minipage}{.9\linewidth}\tt \small}
{\end{minipage}\hspace*{\fill}\bigskip}

\newcommand{\bS}{\textrm{\bf S}}
\newcommand{\bT}{\textrm{\bf T}}
\newcommand{\cM}{\mathcal{M}}
\newcommand{\Cons}{\underline{\mathsf{Cons}}}
\newcommand{\Vars}{\underline{\mathsf{Nulls}}}
\newcommand{\Terms}{\underline{\mathsf{Terms}}}
\newcommand{\PTerms}{\underline{\mathsf{PTerms}}}
\newcommand{\Sol}{\textrm{\sf Sol}}
\newcommand{\dom}{\text{\rm dom}}
\newcommand{\FO}{\text{\rm FO}\xspace}
\newcommand{\FOO}{\text{\rm FO$^<$}\xspace}
\newcommand{\FOOO}{\text{\rm FO$^{(<)}$}\xspace}
\newcommand{\CQ}{\text{\rm CQ}\xspace}
\newcommand{\CQO}{\text{\rm CQ$^<$}\xspace}
\newcommand{\UCQ}{\text{\rm UCQ}\xspace}
\newcommand{\UCQO}{\text{\rm UCQ$^<$}\xspace}

\begin{document}

 \author{Balder ten Cate\inst{1} \and Laura Chiticariu\inst{2} \and
 Phokion Kolaitis\inst{3} \and Wang-Chiew Tan\inst{4}}

\institute{University of Amsterdam \and IBM Almaden \and UC Santa Cruz and IBM Almaden \and UC Santa Cruz}

\title{Laconic schema mappings: computing core universal solutions by
  means of SQL queries\thanks{This work was carried out during a visit of the
    first author to UC Santa Cruz and IBM Almaden. The work of the first
    author funded by the Netherlands Organization of Scientific
    Research (NWO) grant      
 639.021.508 and NSF grant IIS-0430994. The work of the
    third author partly funded by NSF grant IIS-0430994. The work of the
    fourth author partly funded by NSF CAREER Award IIS-0347065 and
    NSF grant IIS-0430994.  }}
\maketitle 

\begin{abstract}
  We present a new method for computing core universal solutions in
  data exchange settings specified by source-to-target dependencies,
  by means of SQL queries. Unlike previously known algorithms, which
  are recursive in nature, our method can be implemented directly on
  top of any DBMS.  Our method is based on the new notion of a laconic
  schema mapping.  A laconic schema mapping is a schema mapping for
  which the canonical universal solution is the core universal
  solution.  We give a procedure by which every schema mapping
  specified by \FO s-t tgds can be turned into a laconic schema
  mapping specified by \FO s-t tgds that may refer to a linear order
  on the domain of the source instance.  We show that our results are
  optimal, in the sense that the linear order is necessary and the
  method cannot be extended to schema mapping involving target
  constraints.
\end{abstract}

\section{Introduction}

  We present a new method for computing core universal solutions in
  data exchange settings specified by source-to-target dependencies,
  by means of SQL queries. Unlike previously known algorithms, which
  are recursive in nature, our method can be implemented directly on
  top of any DBMS.  Our method is based on the new notion of a laconic
  schema mapping.  A laconic schema mapping is a schema mapping for
  which the canonical universal solution is the core universal
  solution.  We give a procedure by which every schema mapping
  specified by \FO s-t tgds can be turned into a laconic schema mapping
  specified by \FO s-t tgds that may refer to a linear order on the
  domain of the source instance.

\paragraph*{Outline of the paper:} In Section~\ref{sec:prel}, we recall
basic notions and facts about schema mappings.
Section~\ref{sec:obtaining-solutions} explains what it means to
compute a target instance by means of SQL queries, and we state our
main result.  Section~\ref{sec:laconicity} introduces the notion of
laconicity, and contains some initial observations In
Section~\ref{sec:transformation}, we present our main result, namely a
method for transforming any schema mapping specified by \FO s-t tgds
into a laconic schema mapping specified by \FO s-t tgds asssuming a
linear order.
In Section~\ref{sec:target-constraints}, we show that our results cannot
be extended to the case with target constraints. 

\section{Preliminaries} \label{sec:prel}

In this section, we recall basic notions from data exchange and fix
our notation.

\subsection{Instances and homomorphisms}
\label{sec:prel-instances}

Fix disjoint infinite sets of constant values $\Cons$ and null values
$\Vars$, and let $<$ be a linear order on $\Cons$. We consider
instances whose values are from $\Cons\cup\Vars$. We use $\dom(I)$ to
denote the set of values that occur in facts in the instance $I$.  A
\emph{homomorphism} $h:I\to J$, with $I,J$ instances of the same
schema, is a function $h:\Cons\cup\Vars\to\Cons\cup\Vars$ with
$h(a)=a$ for all $a\in\Cons$, such that for all relations $R$ and all
tuples of (constant or null) values $(v_1,\ldots,v_n)\in R^I$,
$(h(v_1),\ldots,h(v_n))\in R^J$.  Instances $I,J$ are
\emph{homomorphically equivalent} if there are homomorphisms $h:I\to
J$ and $h':J\to I$.  An \emph{isomorphism} $h:I\cong J$ is a
homomorphism that is a bijection between $\dom(I)$ and $\dom(J)$
and that preserves truth of atomic formulas in both directions.
Intuitively, nulls act as placeholders for actual (constant) values,
and a homomorphism from $I$ to $J$ captures the fact that $J$
``contains more, or at least as much information'' as $I$.

The \emph{fact graph} of an instance $I$ is the graph whose nodes are
the facts $R\vec{v}$ (with $R$ a $k$-ary relation and
$\vec{v}\in(\Cons\cup\Vars)^k$, $k\geq 0$) true in $I$, and such that there is an
edge between two facts if they have a null value in common.

We will denote by \CQ, \UCQ, and \FO the set of conjunctive queries,
unions of conjunctive queries, and first-order queries, respectively,
and \CQO, \UCQO, and \FOO are defined similarly, except that the
queries may refer to the linear order.  Thus, unless
indicated explicitly, it is assumed that queries do not refer to the
linear order. For any query $q$ and instance $I$, we denote by $q(I)$
the answers of $q$ in $I$, and we denote by $q(I)_{\downarrow}$ the
ground answers of $q$, i.e., $q(I)_{\downarrow}=q(I)\cap\Cons^k$ for
$k$ the arity of $q$.

\subsection{Schema mappings, universal solutions, and certain answers}

Let $\bS$ and $\bT$ be disjoint schemas, called the source schema and
the target schema.  As usual in data exchange, whenever we speak of a
source instance, we will mean an instance of $\bS$ whose values belong
to $\Cons$, and when we speak of a target instance, we will mean a
instance of $\bT$ whose values may come from $\Cons\cup\Vars$.

A schema mapping is a triple
$\cM=(\mathbf{S},\mathbf{T},\Sigma_{st})$, where $\mathbf{S}$ and
$\mathbf{T}$ are the source and target schemas and $\Sigma_{st}$ is a
finite set of sentences of some logical language defining a class of
pairs of instances $\langle I,J\rangle$. Here, $\langle I,J\rangle$
denotes union of a source instance $I$ and a target instance $J$,
which is itself an instance over the joint schema $\bS\cup\bT$, and
the logical languages we consider are presented below.  Two schema
mappings, $\cM=(\bS,\bT,\Sigma_{st})$ and
$\cM'=(\bS,\bT,\Sigma'_{st})$, are said to be \emph{logically
  equivalent} if $\Sigma_{st}$ and $\Sigma'_{st}$ are logically
equivalent, i.e., satisfied by the same pairs of instances.  Given a
schema mapping $\cM=(\mathbf{S},\mathbf{T},\Sigma_{st})$ and a source
instance $I$, a \emph{solution} for $I$ with respec to $\cM$ is a target
instance $J$ such that $\langle I,J\rangle$ satisfies $\Sigma_{st}$.
We denote the set of solutions for $I$ with respect to $\cM$ by
$\Sol_{\cM}(I)$, or simply $\Sol(I)$ when the schema mapping is clear
from the context.

The concrete logical languages that we will consider for the
specification of $\Sigma_{st}$ are the following. A
\emph{source-to-target tuple generating dependency} (s-t tgd) is a
first-order sentence of the form
$\forall\vec{x}(\phi(\vec{x})\to\exists\vec{y}.\psi(\vec{x},\vec{y}))$,
where $\phi$ is a conjunction of atomic formulas over $\bS$ and $\psi$
is a conjunction of atomic formulas over $\bT$, such that each
variable in $\vec{x}$ occurs in $\phi$.  A more general class of
constraints called \emph{\FO s-t tgds} is defined analogously, except
that the antecedent is allowed to be any \FO query over
$\bS$. Similarly, \emph{$L$ s-t tgds} can be defined for any query
language $L$. A \emph{LAV s-t tgd}, finally, is an s-t tgd in which
$\phi$ is a single atomic formula. To simplify notation, we will
typically drop the universal quantifiers when writing ($L$) s-t tgds.

Given a source instance $I$, a schema mapping $\cM$, and a target
query $q$, we will denote by $certain_{\cM,q}(I)$ the set of
\emph{certain answers} to $q$ in $I$ with respect to $\cM$, i.e., the
intersection $\bigcap_{J\in \Sol_{\cM}(I)} q(J)$. In other words, a
tuple of values is a certain answer to $q$ if belongs to the set of
answers of $q$, no matter which solution of $I$ one picks.  There are
two methods to compute certain answers to a conjunctive query. The
first method uses \emph{universal solutions} and the second uses
\emph{query rewriting}.

A universal solution for a source instance $I$ with respect to a
schema mapping $\cM$ is a solution $J\in \Sol_\cM(I)$ such that, for
every $J'\in \Sol_\cM(I)$, there is a homomorphism from $J$ to $J'$.
It was shown in \cite{Fagin05:data} that the certain answers for a
conjunctive target query can be obtained simply by evaluating the
query in a universal solution. Moreover, universal solutions are
guaranteed to exist for schema mappings specified by $L$ s-t tgds, for
\emph{any} query language $L$.

\begin{theorem}[\cite{Fagin05:data}]
  For all schema mappings $\cM$, source instances
  $I$, conjunctive queries $q$, and universal solutions $J\in
  \Sol(I)$,  $certain_{\cM}(q)(I) = q(J)_{\downarrow}$.
\end{theorem}

\begin{theorem}[\cite{Fagin05:data}] \label{thm:existence-usol} For
  every schema mapping $\cM$ specified by $L$ s-t tgds, with $L$ any
  query language, and for every source instance $I$, there is a
  universal solution for $I$ with respect to $\cM$.
\end{theorem}

Theorem~\ref{thm:existence-usol} was shown in \cite{Fagin05:data} for
schema mappings specified by s-t tgds but the same argument applies to
schema mappigns specified by $L$ s-t tgds, for any query language $L$. 
We will discuss concrete methods for constructing universal solutions in
Section~\ref{sec:obtaining-solutions}.

The second method for computing certain answers to conjunctive queries
is by rewriting the given target query to a query over the source
that directly computes the certain answers to the original query.

\begin{theorem}\label{thm:certain} Let $L$ be any of $\UCQ$, $\UCQO$, $\FO$, $\FOO$. Then
 for every schema mapping $\cM$ specified by s-t tgds and for every
 $L$-query $q$ over $\bT$, one can compute in exponential
  time an $L$-query over $\bS$ defining $certain_{\cM,q}$.
\end{theorem}

There are various ways in which such certain answer queries can be
obtained. One possibility is to split up the schema mapping $\cM$ into
a composition $\cM_1\circ\cM_2$, with $\cM_1$ specified by full
s-t tgds and $\cM_2$ specified by LAV s-t tgds, and then to
successively apply the known query rewriting techniques of MiniCon
\cite{Pottinger01:minicon} and full s-t tgd unfolding (cf.~\cite{Lenzerini02:data}). In~\cite{CateKolaitis}, an alternative
rewriting method was given for the case of $L=\FOOO$, which can be used to compute in 
polynomial time an $\FOOO$
source query $q'$ defining $certain_{\cM,q}$ 
over source instances whose domain contains at least two elements.

\subsection{Core universal solutions}\label{sec:cores}

A source instance can have many universal solutions. Among these, the
\emph{core universal solution} plays a special role.  A target
instance $J$ is said to be a \emph{core} if there is no proper
subinstance $J'\subseteq J$ and homomorphism $h:J\to J'$.  There is
equivalent definition in terms of retractions. A subinstance
$J'\subseteq J$ is called a \emph{retract} of $J$ if there is a
homomorphism $h:J\to J'$ such that for all $a\in dom (J')$, $h(a)=a$.
The corresponding homomorphism $h$ is called a \emph{retraction}. A
retract is \emph{proper} if it is a proper subinstance of the original
instance. A core of a target instance $J$ is a retract of $J$ that has
itself no proper retracts.  Every (finite) target instance has a
unique core, up to isomorphism. Moreover, two instances are
homomorphically equivalent iff they have isomorphic cores. It follows
that, for every schema mapping $\cM$, every source instance has at
most one core universal solution up to isomorphism. Indeed, if the
schema mapping $\cM$ is specified by \FO s-t tgds then each source
instance has \emph{exactly} one core universal solution up to
isomorphism \cite{Fagin05:core}. We will therefore freely speak of \emph{the}
core universal solution.

It has been convincingly argued that, among all universal solutions
for a source instance, the core universal solution is the preferred
solution. One important reason is that the core universal solution is
the smallest universal solution: if $J$ is the core universal solution
for a source instance $I$ with respect to a schema mapping $\cM$, and
$J'$ is any other solution universal solution for $I$, i.e., one that
is not a core, then $|J|<|J'|$. Consequently, the core universal
solution is the universal solution that is least expensive to
materialize. We add to this a second virtue of the core universal
solution, namely that, among all universal solutions, it is the most
conservative one in terms of the answers that it assigns to
conjunctive queries with inequalities.  We propose another reason to
be interested in the core universal solution, namely that it is the
solution that satisfies the most dependencies. In many practical data
exchange settings, one is interested in solutions satisfying certain
target dependencies. One way to obtain such solutions is to include
the relevant target dependencies in the specification of the schema
mapping. If the target dependencies satisfy certain syntactic
requirements (in particular, if they form a weakly acyclic set of
target tgds and target egds), then a solution satisfying these target
dependencies can be obtained by means of the chase. On the other hand,
sometimes it happens that the universal solution one constructs
without taking into account the target dependencies happens to satisfy
the target dependencies. Whether this happens depends very much on
which universal solution is constructed. For example if $\cM$ is the
schema mapping specified by the s-t tgd $Rx \to \exists y.Syx$, $I$ is
any source instance and $J$ a universal solution, then the first
attribute of $S$ is not necessarily a key in $J$. However, if $J$ is
the core universal solution, then it \emph{will} be a key.  In fact,
it turns out that the core universal solution is the universal
solution that maximizes the set of valid target dependencies.
To make this precise, let a \emph{disjunctive target dependency} be a
first-order sentence of the from
$\forall\vec{x}\phi(\vec{x})\to\bigvee_i\exists\vec{y}_i.\psi_i(\vec{x},\vec{y}_i))$,
where $\phi, \psi_i$ are conjunctions of atomic formulas over the
target schema $\bT$
and/or equalities.  Then we have:

 \begin{theorem} Let $\cM$ be any schema mapping,
   $I$ be any source instance, $J$ the core universal solution of
   $I$, and $J'$ any other universal solution of $I$, i.e., one that 
  is not a core. Then
  \begin{enumerate}
 \item Every disjunctive dependency valid on $J'$ is valid on $J$, and
 \item Some disjunctive dependency valid on $J$ is not valid on $J'$.
  \end{enumerate}
\end{theorem}

\begin{proof}
  The first half of the result follows from the fact that $J$ is a
  retract of $J'$ and disjunctive dependencies are preserved when
  going from an instance to one of its retract. This is shown in
  \cite{GottlobNash08:core} for non-disjunctive embedded dependencies, but
  the same argument applies to disjunctive dependencies. 
  To prove the second half, pick fresh variables $\vec{x}$, one for
  each value (constant or null) in the domain of $J'$, and let
  $\psi(\vec{x})$ be the conjunction of all facts that are true in
  $J'$ under the natural assignment.  Consider the disjunctive
  dependency
  $\forall \vec{x}(\psi(\vec{x}) \to \bigvee_{i\neq j}(x_i=x_j))$.
  This disjunctive dependency is clearly not true in $J'$ but it
  is trivially true in $J$, since $J$, being a proper retract of
  $J'$, contains strictly fewer nulls than $J'$. \qed
\end{proof}

Concerning the complexity of computing core
universal solutions, we have the following:

\begin{theorem}[\cite{Fagin05:core}]
  For fixed schema mappings specified by \FOO s-t tgds, given a source
  instance, a core universal solution can be computed in polynomial
  time.
\end{theorem}

Strictly speaking, in \cite{Fagin05:core} this was only shown for schema mappings
specified by s-t tgds. However the same argument applies to \FOO s-t tgds.  In
fact, this holds for richer data exchange settings, there the schema
mapping specification may contain also target constraints
(specifically, \emph{target egds} and \emph{weakly acyclic target
  tgds}). Moreover, several algorithms for obtaining core universal
solutions in polynomial time have been proposed.  %

\section{Computing universal solutions with SQL queries}
\label{sec:obtaining-solutions}

There is a discrepancy between the methods for computing universal
solutions commonly presented in the data exchange literature, and the
methods actually employed by data exchange tools. In the data exchange
literature, methods for computing universal solutions are often
presented in the form of a chase procedures. In practical
implementations such as Clio, on the other hand, it is common to
compute universal solutions using SQL queries, thus leveraging the
capabilities of existing DBMSs. We briefly review here both
approaches, and explain how canonical universal solutions can be computed
using SQL queries.

The simplest and most well known method for computing universal
solutions is the \emph{naive chase}\footnote{There are also other, more
sophisticated versions of the chase, but they will not be relevant for
most of what we discuss, since we will be interested in computing solutions
by means of SQL queries anyway. We \emph{will} briefly mention one variant of
the chase later on.}
The algorithm is described in
Figure~\ref{fig:naive-chase}.  For a source instance $I$ and schema
mapping $\cM$ specified by \FOOO s-t tgds, the result of applying the
naive chase is called the \emph{canonical universal solution} of $I$
with respect to $\cM$.  Note that the result of the naive chase is
unique up to isomorphism, since it depends only on the exact choice of
fresh nulls. Also note that, even if two schema mappings are logically
equivalent, they may assign different canonical universal solutions to
a source instance.
  \begin{figure}[t]
  \textsc{Input:} A schema mapping $\cM=(\bS,\bT,\Sigma_{st})$ and a source instance $I$

  \textsc{Output:} A target instance $J$ that is a universal solution for $I$ w.r.t.~$\cM$
  \bigskip

  \begin{algorithmic}
  \STATE $J := \emptyset$;
  \FORALL{ $\forall\vec{x}(\phi(\vec{x})\to\exists\vec{y}.\psi(\vec{x},\vec{y}))\in \Sigma_{st}$}
  \FORALL{tuples of constants $\vec{a}$ such that $I\models\phi(\vec{a})$}
  \STATE pick a fresh null value $N_i$ for each $y_i$ and add the facts
    in $\psi(\vec{a},\vec{N})$ to $J$
  \ENDFOR
  \ENDFOR;
  \RETURN $J$
  \end{algorithmic}
  \caption{Naive chase method for computing universal solutions.}
  \label{fig:naive-chase}
  \end{figure}
  We will now discuss how canonical universal solutions can be
  equivalently computed by means of SQL queries. The idea is very
  simple, and before giving a rigorous presentation, we illustrate it
  by an example.
  Consider the schema mapping specified by the s-t tgds
  \[\begin{array}{lll}
    Rx_1x_2 &\to& \exists y.(Sx_1 y \land Tx_2 y) \\
    Rxx     &\to& Sxx 
  \end{array}\]
 We first Skolemize the dependencies, and split
  them so that the right hand side consists of a single conjunct. In 
  this way, we get
   \[\begin{array}{lll}
   Rx_1x_2 &\to& Sx_1f(x_1,x_2) \\
   Rx_1x_2 &\to& Tx_2f(x_1,x_2) \\
   Rxx &\to& Sxx 
 \end{array}\]
 Next, for each target relation $R$ we collect the dependencies that
 contain $R$ in the right hand side, and we interpret these as
 constituting a definition of $R$. In this way, we get the following definitions of 
 $S$ and $T$:

 \[\begin{array}{lll}
   S &:=& \{(x_1,f(x_1,x_2))\mid Rx_1x_2\} \cup \{(x,x)\mid Rxx\} \\
   T &:=& \{(x_2,f(x_1,x_2))\mid Rx_1x_2\} 
  \end{array}\]
  In general, the definition of a $k$-ary target relation $R\in\bT$ will be of the shape
\begin{equation} \label{eq:interpretation}
  R ~:=~ \{(t_1(\vec{x}), \ldots, t_k(\vec{x}))\mid\phi(\vec{x})\}
  ~\cup~ \cdots ~\cup~ \{(t'_1(\vec{x}'), \ldots, t'_k(\vec{x}))\mid\phi'(\vec{x}')\}
\end{equation}
where $t_1, \ldots, t_k, \ldots, t'_1, \ldots, t'_k$ are terms and
$\phi, \ldots, \phi'$ are first-order queries over the source schema.
Since \FO queries correspond to SQL queries, one can easily use a
relational DBMS in order to compute the tuples in the relation $R$.

  The general idea behind the construction of the \FO queries
  should be clear from the example. However, giving a precise
  definition of what it means to compute a target instance by means of
  SQL queries require a bit of care. We need to assume some structure
  on the set of nulls $\Vars$.  Fix a countably infinite set of
  function symbols of arity $n$, for each $n\geq 0$. For any set $X$,
  denote by $\Terms[X]$ be the set of all \emph{terms} built up from
  elements of $X$ using these function symbols, and denote by
  $\PTerms[X]\subseteq \Terms[X]$ the set of all \emph{proper terms},
  i.e., those with at least one occurrence of a function symbol. For
  instance, if $g$ is a unary function and $h$ is a binary function,
  then $h(g(x),y)$, $g(x)$ and $x$ belong to $\Terms[\{x,y\}]$, but
  only the first two belong to $\PTerms[\{x,y\}]$.  It is important to
  distinguish between proper terms built up from constants on the one
  hand and constants on the other hand, as the former will be treated
  as nulls and the latter not. More precisely, we assume that
  $\PTerms[\Cons]\subseteq\Vars$. Recall that
  $\Cons\cap\Vars=\emptyset$.

\begin{definition}[$L$-term interpretation]
  Let $L$ be any query language. An \emph{$L$-term interpretation} $\Pi$ is a map
  assigning to each $k$-ary relation symbol $R\in\bT$ a union of
  expressions of the form (\ref{eq:interpretation}) where $t_1,
  \ldots, t_k \in \Terms[\vec{x}]$ and $\phi(\vec{x})$ is an $L$-query
  over $\bS$.
\end{definition}

Given a source instance $I$, an $L$-term interpretation $\Pi$ generates an
target instance $\Pi(I)$, in the obvious way. Note that $\Pi(I)$ may
contain constants as well as nulls. Although the program specifies
exactly which nulls are generated, we will consider $\Pi(I)$ only up
to isomorphism, and hence the meaning of an $L$-term interpretation does
not depend on exactly which function symbols it uses.

The previous example shows

\begin{proposition} \label{prop:canon-FO}
  Let $L$ be any query language.
  For every schema mapping specified by $L$ s-t tgds there is
  an $L$-term interpretation that yields for each source instance the canonical
  universal solution.
\end{proposition}

Incidentally, even for schema mappings specified by SO tgds, 
as defined in \cite{Fagin05:composing},  $\FO$-term interpretations can be 
constructed that compute the canonical universal solution. However, 
the above suffices for present purposes. 

On the other hand,

\begin{proposition} \label{prop:no-core-program} No \FO-term
  interpretation yields for each source instance the core
  universal solution with respect to the schema mapping specified
  by the \FO (in fact LAV) s-t tgd $Rxy\to \exists z.(Sxz\land Syz)$.
\end{proposition}

\begin{proof}The argument uses the fact that \FO formulas are
  invariant for automorphisms.  Let $I$ be the source instance whose
  domain consists of the constants $a,b,c,d$, and such that $R$ is the
  total relation over this domain. Note that every permutation of the
  domain is an automorphism of $I$. Suppose for the sake of
  contradiction that there is an \FO-term interpretation $\Pi$ such
  that the $\Pi(I)$ is the core universal solution of $I$. Then the
  domain of $\Pi(I)$ consists of the constants $a,b,c,d$ and a
  distinct null term, call it $N_{\{x,y\}}\in\PTerms[\vec{x}]$, for
  each pair of distinct constants $x,y\in\{a,b,c,d\}$, and $\Pi(I)$
  contains the facts $RxN_{\{x,y\}}$ and $RyN_{\{x,y\}}$ for each of
  these nulls $N_{\{x,y\}}$.  Now consider the term $N_{\{a,b\}}$. We
  can distinguish two cases. The first case is where the term
  $N_{\{a,b\}}$ does not contain any constants as arguments.  In this
  case, it follows from the invariance of \FO formulas for
  automorphisms that $\Pi(I)$ contains $RxN_{\{a,b\}}$ for every
  $x\in\{a,b,c,d\}$, which is clearly not true. The second case is
  where $N_{\{a,b\}}$ contains at least one constant as an
  argument. If $N_{\{a,b\}}$ contains the constant $a$ or $b$ then let
  $t'$ be obtained by switching all occurrences of $a$ and $b$ in
  $N_{\{a,b\}}$, otherwise let $t'$ be obtained by switching all
  occurrences of $c$ and $d$ in $N_{\{a,b\}}$. Either way, we obtain
  that there is a second null, namely $t'$, which is distinct from
  $N_{\{a,b\}}$, and which stands in exactly the same relations to $a$
  and $b$ as $N_{\{a,b\}}$ does. This again contradicts our assumption
  that $J$ is the core universal solution of $I$.
\end{proof}

Things change in the presence of a linear order. We will show that
every schema mapping specified by \FOO s-t tgds is logically
equivalent to a laconic schema mapping specified by \FOO s-t tgds,
i.e., one for which the canonical universal solution is always a
core. In particular, given Proposition~\ref{prop:canon-FO}, this
shows:

\begin{theorem}\label{thm:core-FO}
  For every schema mapping specified by \FOO s-t tgds there is a
  \FOO-term interpretation that yields for each source instance the core
  universal solution.
\end{theorem}

In the case of the example from
Proposition~\ref{prop:no-core-program},
the \FOO-term interpretation $\Pi$ computing the core universal
solution
is given by 
\[\begin{split}\Pi(S) = \{(x_1,f(x_1,x_2)) \mid (Rx_1x_2 \lor Rx_2x_1)
\land x_1 \leq x_2 \} \\ {~\cup~} \{(x_2,f(x_1,x_2)) \mid (Rx_1x_2 \lor Rx_2x_1)
\land x_1 \leq x_2 \}\end{split}\]

Furthermore, we will show that every schema mapping defined by 
\FO s-t tgds whose right-hand-side contains at most one atomic 
formula is equivalent to a laconic schema mapping specified by
\FO s-t tgds, and therefore, its core universal solutions can be computed
by means of an \FO-term interpretation. 
In other words, in this case the linear order is not needed. Note that
in the example from Proposition~\ref{prop:no-core-program}, the right-hand-size of the 
s-t tgd consists of two atomic formulas.

In the next section, we formally introduce the notion of laconicity.
In Section~\ref{sec:transformation}, we show that every schema mapping
specified by \FOO s-t tgds is logically equivalent to a laconic schema
mapping specified by \FOO s-t tgds.

\section{Laconicity}\label{sec:laconicity}

A schema mapping is laconic if the canonical universal solution of a
source instance coincides with the core universal solution.  In
particular, for laconic schema mappings the core universal solution
can be computed using any method for computing canonical universal
solutions, such as the ones described in
Section~\ref{sec:obtaining-solutions}. In this section, we discuss
some examples and general observations concerning laconicity, in order
to make the reader familiar with the notion.  In the next section we
will focus on constructing laconic schema mappings.  In particular, we
will show there that every schema mapping specified by \FOO s-t tgds
is logically equivalent to a laconic schema mapping specified by \FOO
s-t tgds.

\begin{definition}[Laconicity]
  A schema mapping is \emph{laconic} if
  for every source instance $I$, the canonical universal solution of
  $I$ with respect to $\cM$ is a core.
\end{definition}

Note that the definition only makes sense for schema mappings
specified by \FOOO s-t tgds, because we have defined the notion of a
canonical universal solution only for such schema mappings.

\begin{figure*}[t]
\begin{center}
\begin{tabular}{l@{\quad}l@{\qquad}l@{\quad}l}
\multicolumn{2}{l}{\emph{Non-laconic schema mapping}} &
\multicolumn{2}{l}{\emph{Logically equivalent laconic schema mapping}} \\
\\
(a) & $Px \to \exists yz.Rxy\land Rxz$ & (a$'$) & $Px \to \exists y.Rxy$ \\
\\
(b) & $Px \to \exists y.Rxy$ & (b$'$) & $Px \to Rxx$ \\
    & $Px \to Rxx$ & \\
\\
(c) & $Rxy \to Sxy$ & (c$'$) & $Rxy \to Sxy$ \\
    & $Px \to \exists y.Sxy$  && $Px\land\neg \exists y.Rxy\to \exists y.Sxy$ \\
\\
(d) & $Rxy\to\exists z.Sxyz$ & (d$'$) & $Rxy\land x\neq y \to \exists z.Sxyz$ \\
    & $Rxx\to Sxxx$ && $Rxx \to Sxxx$ \\
\\
(e) & $Rxy\to \exists z.(Sxz\land Syz)$ & (e$'$) & $(Rxy\lor Ryx) \land
x\leq y \to \exists z.(Sxz\land Syz)$ 
\end{tabular}
\end{center}
\caption{Examples of non-laconic schema mappings and their laconic
  equivalents.} \label{fig:examples}
\end{figure*}

Examples of laconic and non-laconic schema mappings are given in
Figure~\ref{fig:examples}.  It is easy to see that every schema
mapping specified by full s-t tgds only (i.e., s-t tgds without
existential quantifiers) is laconic. Indeed, in this case, the
canonical universal solution does not contain any nulls, and hence is
guaranteed to be a core. Thus, being specified by full s-t tgds is a
sufficient condition for laconicity, although a rather uninteresting
one. The following provides us with a necessary condition, which
explains why the schema mapping in Figure~\ref{fig:examples}(a) is not
laconic. Given an s-t tgd
$\forall\vec{x}(\phi\to\exists\vec{y}.\psi)$, by \emph{the canonical
  instance of $\psi$}, we will mean the target instance whose facts
are the conjuncts of $\psi$, where the $\vec{x}$ variables are treated
as constants and the $\vec{y}$ variables as nulls.

\begin{proposition}
  If a schema mapping $(\bS,\bT,\Sigma_{st})$ specified by s-t tgds is
  laconic, then for each s-t tgd
  $\forall\vec{x}(\phi\to\exists\vec{y}.\psi)\in\Sigma_{st}$, the
  canonical instance of $\psi$ is a core.
\end{proposition}

\begin{proof}
  We argue by contraposition. Suppose the canonical instance $J$ of
  $\psi$ is not a core.  Let $J'\subseteq J$ be the core of $J$ and
  $h:J\to J'$ the corresponding retraction.

  Take any source instance $I$ in which $\phi$ is satisfied under an
  assignment $g$, and let $K$ be the canonical universal solution of
  $I$.  Since $\phi$ is true in $I$ under the assignment $g$ and by
  the construction of the canonical universal solution, we have that
  $g$ extends to a homomorphism $\widehat{g}:J\to K$ sending the
  $\vec{y}$ values to disjoint nulls. In fact, we may assume without
  loss of generality that $\widehat{g}(y_i)=y_i$ for each $y_i\in
  \vec{y}$. Moreover, by the construction of canonical universal
  solutions these null values will not play any further role in
  subsequent steps of the chase. In particular, they do not
  participate in any facts of $K$ other than those in the
  $\widehat{g}$-image of $J$. By the $\widehat{g}$-image of $J$
  we mean the subinstance of $K$ containing those facts that are
  in the image of the homomorphism $\widehat{g}:J\to K$. 

  Finally, let $K'$ be the subinstance of $K$ in which the
  $\widehat{g}$-image of $J$ is replaced by the $\widehat{g}$-image 
  of $J'$. Then $h:J\to J'$ naturally extends to a homomorphism
  $h':K\to K'$. Since $K'$ is a proper subinstance of $K$, we conclude
  that $K$ is not a core, and therefore, $\cM$ is not laconic.\qed
\end{proof}

In the case of schema mapping (e) in Figure~\ref{fig:examples}, the
linear order is used in order to obtain a logically equivalent
laconic schema mapping (e$'$). Note that the schema mapping (e$'$) is
\emph{order-invariant} in the sense that the set of solutions of a
source instance $I$ does not depend on the interpretation of the $<$
relation in $I$, as long as it is a linear order. Still, the use of the
linear order cannot be avoided, as follows from
Proposition~\ref{prop:no-core-program}.  What is really going on, in
this example, is that the right hand side of (e) has a non-trivial
automorphism (viz.~the map sending $x$ to $y$ and vice versa), and the
conjunct $x\leq y$ in the antecedent of (e$'$) plays, intuitively, the
role of a tie-breaker, cf.~Section~\ref{sec:side-conditions}.

Testing whether a given schema mapping is laconic is not a tractable
problem:

\begin{proposition}\label{prop:complexity-laconicity}
  Testing laconicity of schema mappings specified by \FO s-t tgds is
  undecidable. It is NP-hard already for schema mappings specified by
  LAV s-t tgds.
\end{proposition}

\begin{proof}
  The first claim is proved by a reduction from the  
  satisfiability problem for first-order logic on finite
  instances, which is undecidable by Trakhtenbrot's theorem.
  For any first-order formula $\phi(\vec{x})$, let $\cM_\phi$
  be the schema mapping containing only one dependency,
  namely $\forall\vec{x}(\phi(\vec{x})\to\exists y_1 y_2.(Py_1\land Py_2))$. 
  It is easy to see that 
  $\cM_\phi$ is laconic iff $\phi$ is not satisfiable. 

  The NP-hardness in the case of LAV mappings is proved by a reduction
  from the core testing problem (given a graph, is it a core), which
  is known to be NP-complete \cite{Hell92:core}. Consider any graph $G=(V,E)$ and let
  $\exists\vec{y}.\phi(\vec{y})$ be the Boolean canonical
  conjunctive query of $G$. Let $\cM_G$ be the schema mapping whose
  only dependency is $\forall x.(Px \to
  \exists\vec{y}.(\phi(\vec{y})\land\bigwedge_i Rxy_i)$.  Then $\cM_G$
  is laconic iff $G$ is a core.
\qed
\end{proof}

\section{Making schema mappings laconic}\label{sec:transformation}

In this section, we present a procedure for transforming any schema
mapping $\cM$ specified by \FOO s-t tgds into a logically equivalent
\emph{laconic} schema mapping $\cM'$ specified by \FOO s-t tgds.  To
simplify the notation, throughout this section, we assume a fixed
input schema mapping $\cM=(\bS,\bT,\Sigma_{st})$, with $\Sigma_{st}$ a
finite set of \FOO s-t tgds. Moreover, we will assume that the \FOO
s-t tgds $\forall\vec{x}(\phi\to\exists\vec{y}.\psi)\in\Sigma_{st}$
are non-decomposable \cite{Fagin05:core}, meaning that the fact graph
of $\exists\vec{y}.\phi(\vec{x},\vec{y})$ (where the facts are the
conjuncts of $\phi$ and two facts are connected if they have an
existentially quantified variable in common) is connected. This
assumption is harmless: every \FOO s-t tgd can be decomposed into a
logically equivalent finite set of non-decomposable \FOO s-t tgds
(with identical left-hand-sides, one for each connected component of
the fact graph) in polynomial time.

The outline of the procedure for making schema mappings laconic is as
follows (the items correspond to subsections of the present section):

\begin{enumerate}
\item Construct a finite list ``fact block types'': 
  descriptions of potential fact blocks in core universal solutions.
\item Compute for each of the fact block types a precondition: a
  first-order formula over the source schema that tells exactly when
  the core universal solution will contain a fact block of the given type.
\item If any of the fact block types has non-trivial
  automorphisms, add an additional side condition, consisting of a
    Boolean combination of formulas of the form $x_i<x_j$, in
  order to avoid that multiple copies of the same fact block
  are created in the canonical universal solution.
\item Construct the new schema mapping $\cM'=(\bS,\bT,\Sigma'_{st})$,
  where $\Sigma'_{st}$ contains an \FOO s-t tgd for each of the
  fact block types. The left-hand-side of the \FOO s-t tgd is
  the conjunction of the precondition and side condition of the
  respective fact block type, while the right-hand-side is
  the fact block type itself.
\end{enumerate}
We illustrate the approach by means of an example. The technical
notions that we use in discussing the example will be formally defined
in the next subsections.

\begin{example}
  Consider the schema mapping $\cM=(\{P,Q\},\{R_1,R_2\},\Sigma_{st})$,
  where $\Sigma_{st}$ consists of the dependencies 
  \[\begin{array}{l}Px \to \exists y.R_1xy  \\
    Qx \to \exists yzu.(R_2xy\land R_2zy\land R_1zu)\end{array}\]
  In this case, there are exactly three relevant fact block
  types. They are listed below, together with their preconditions.
  \[\begin{array}{lll@{\quad}l}
    \multicolumn{3}{l}{\textit{Fact block type}} & \textit{Precondition} \smallskip\\
    t_1(x;y) &=& \{R_1xy\} & pre_{t_1}(x) = Px \\
    t_2(x;yzu) &=& \{R_2xy, R_2zy, R_1zu\} & pre_{t_2}(x) = Qx\land\neg Px \\
    t_3(x;y) &=& \{R_2xy\} & pre_{t_3}(x) = Qx\land Px
  \end{array}\]
  We use the notation $t(\vec{x};\vec{y})$ for a fact block type to
  indicate that the variables $\vec{x}$ stand for constants and the
  variables $\vec{y}$ stand for distinct nulls.  

  As it happens, the above fact block types have no
  non-trivial automorphisms. Hence, no side conditions need to be
  added, and $\Sigma'_{st}$ will consist of the following \FO s-t
  tgds:
  \[\begin{array}{lll}
    Px &\to& \exists y.R_1xy  \\
    Qx\land\neg Px &\to& \exists yzu.(R_2xy\land R_2zy\land R_1zu) \\
    Qx\land Px &\to& \exists y.(R_2xy) 
   \end{array}\]
The reader may verify that in this case, the obtained schema mapping
is indeed laconic. We will prove in Section~\ref{sec:putting} that
the output of our transformation is guaranteed to be a laconic schema
mapping that is logically equivalent to the input schema mapping.
\hfill$\dashv$
\end{example}

We will now proceed to define all the notions appearing in this
example.

\subsection{Generating the fact block types}

Recall that the fact graph of an instance is the graph
whose nodes are the facts of the instance, and such that there is an
edge between two facts if they have a null value in common. A
\emph{fact block}, or \emph{f-block} for short, of an instance
is a connected component of the fact graph of the instance. We
know from \cite{Fagin08:towards} that, for any schema
mapping $\cM$ specified by \FOO s-t tgds, the size of f-blocks in core
universal solutions for $\cM$ is bounded. Consequently, there is
a finite number of f-block types, such that every core universal solution
consists of f-blocks of these types. This is a crucial observation
that we will exploit in our construction.

Formally, an \emph{f-block type} $t(\vec{x};\vec{y})$ will be a finite
set of atomic formulas in $\vec{x}, \vec{y}$, where $\vec{x}$ and
$\vec{y}$ are disjoint sets of variables. We will refer to $\vec{x}$
as the \emph{constant variables} of $t$ and $\vec{y}$ as the
\emph{null variables}. We say that an f-block type
$t(\vec{x};\vec{y})$ is a renaming of an f-block type
$t'(\vec{x}';\vec{y}')$ if there is a bijection $f$ between $\vec{x}$
and $\vec{x}'$ and between $\vec{y}$ and $\vec{y}'$, such that $t'=\{
R(f(\vec{v})) \mid R(\vec{v})\in t\}$. In this case, we write
$f:t\cong t'$ and we call $f$ also a renaming. We will not distinguish
between f-block types that are renamings of each other.  We say that
an f-block $B$ has type $t(\vec{x};\vec{y})$ if $B$ can be obtained
from $t(\vec{x};\vec{y})$ by replacing constant variables by constants
and null variables to distinct nulls, i.e., if $B=t(\vec{a},\vec{N})$
for some sequence of constants $\vec{a}$ and sequence of distinct
nulls $\vec{N}$. Note that we require the relevant substitution to be
injective on the null variables but not necessarily on the constant
variables. If a target instance $J$ contains a block
$B=t(\vec{a},\vec{N})$ of type $t(\vec{x};\vec{y})$ then we say that
$t(\vec{x};\vec{y})$ is \emph{realized} in $J$ \emph{at} $\vec{a}$.
Note that, in general, an f-block type may be realized more than
once at a tuple of constants $\vec{a}$, but this will not happen if
the target instance $J$ is a core universal solution.

We are interested in the f-block types that may be realized in core
universal solutions. Eventually, the schema mapping $\cM'$ that we
will construct from $\cM$ will contain an \FOO s-t tgd for each
relevant f-block type. Not every f-block type as defined above can be
realized.  We may restrict attention to a subclass. Below, by the
canonical instance of an f-block type $t(\vec{x};\vec{y})$ we will
mean the instance containing the facts in $t(\vec{x};\vec{y})$,
considering $\vec{x}$ as constants and $\vec{y}$ as nulls.

\begin{definition}\label{def:plausible}
  The set $\textsc{Types}_\cM$ of f-block types generated by $\cM$
  consists of all f-block types $t(\vec{x};\vec{y})$ satisfying the
  following conditions:
\begin{enumerate}[(a)]
\item $\Sigma_{st}$ contains an \FOO s-t tgd
  $\forall\vec{x}'(\phi(\vec{x}')\to\exists\vec{y}'.\psi(\vec{x}',\vec{y}'))$
  with $\vec{y}\subseteq \vec{y}'$, and $t(\vec{x},\vec{y})$ is the set of conjuncts of
  $\psi$ in which the variables $\vec{y}'-\vec{y}$ do not occur;
\item The canonical instance of $t(\vec{x},\vec{y})$ is a core;
\item The fact graph of the canonical instance of
  $t(\vec{x};\vec{y})$ is connected.
\end{enumerate}
If some f-block types generated by $\cM$ are
renamings of each other, we add only one of them to
$\textsc{Types}_\cM$.
\end{definition}

  The main result of this subsection is:

\begin{proposition}\label{prop:plausible}
  Let $J$ be a core universal solution of a source instance $I$ with
  respect to~$\cM$. Then each f-block of $J$ has type $t(\vec{x};\vec{y})$
  for some $t(\vec{x};\vec{y})\in\textsc{Types}_\cM$.
\end{proposition}

\begin{proof} Let $B$ be any f-block of $J$. 
  Since $J$ is a core universal solution, it is, up to
  isomorphism, an induced subinstance of the canonical universal
  solution $J'$ of $I$. It follows that $J'$ must have an f-block $B'$
  such that $B$ is the restriction of $B'$ to domain of $J$. Since
  $B'$ is a connect component of the fact graph of $J'$, it must
  have been created in a single step during the naive chase. In other 
  words,
  there is an \FOO s-t tgd
  $$\forall\vec{x}(\phi(\vec{x})\to\exists\vec{y}.\psi(\vec{x},\vec{y}))$$
  and an assignment $g$ of constants to the variables $\vec{x}$ and
  distinct nulls to the variables $\vec{y}$ such that $B'$ is
  contained in the set of conjuncts of
  $\psi(g(\vec{x}),g(\vec{y}))$. Moreover, since we assume the 
  \FOO s-t tgds of $\cM$ to be non-decomposable and $B'$ is a 
  a connected component of the fact graph of $J$, $B'$ must be
  exactly the set of facts listed in  $\psi(g(\vec{x}),g(\vec{y}))$.
  In other words, if we let $t(\vec{x};\vec{y})$ be the set of all 
    facts listed in $\psi$, then $B'$ has type $t(\vec{x};\vec{y})$. 
  Finally, let $t'(\vec{x}';\vec{y}')\subseteq t(\vec{x};\vec{y})$ be the 
  set of all facts from $t(\vec{x};\vec{y})$ containing only
  variables $y_i$ for which $g(y_i)$ occurs in $B$. Since $B$ is 
  the restriction of $B'$ to the 
      domain of $J$, we have that $B$ is of type $t'(\vec{x}',\vec{y}')$. 
  Moreover, the fact graph of the canonical instance of $J$ is connected because $B$ is 
  connected, and the canonical instance of $t'(\vec{x}';\vec{y}')$
  is a core, because, if it would not be, then $B$ would not be a
  core either, and hence $J$ would not be a core either, which
  would lead to a contradiction. It follows that $t'(\vec{x}';\vec{y}')\in \textsc{Types}_\cM$.
  \qed
\end{proof}

Note that $\textsc{Types}_\cM$ contains only finitely many f-block
types. Still, the number is in general exponential in the size of the
schema mapping, as the following example shows.

\begin{example}\label{ex:exponentially-many-patterns}
  Consider the schema mapping specified by the following s-t tgds:
  \[\begin{array}{l@{~}ll}
     P_i x &\to P'_i x & \text{(for each $1\leq i\leq k$)} \\
     Q x   &\to \exists y_0y_1\ldots y_k(Rxy_0\land\bigwedge_{1\leq i\leq k}(Ry_i y_0\land P'_i y_i)) \\
   \end{array}\]
   For each $S\subseteq \{1, \ldots, k\}$, the f-block type
   $$t_S(x;(y_i)_{i\in S\cup\{0\}})=\{Rxy_0\}\cup\{Ry_i y_0, P'_i y_i
   \mid i\in S\}$$ belongs to $\textsc{Types}_\cM$. 
   Indeed, each of these
   $2^k$ f-block types is realized in the core
   universal solution of some source instance. 
  The example can be modified to use a fixed
  schemas: replace $P'_i x$ by $Sxx_1\land Sx_1x_2 \land \ldots
  Sx_{i-1}x_i\land Sx_ix_i$. \mbox{}\hfill$\dashv$
\end{example}

The same example %
can be used
 to show
that the smallest logically equivalent schema mapping that is laconic can be exponentially longer. 

\subsection{Computing the precondition of an f-block type}

Recall that, to simplify notation, we assume a fixed schema mapping
$\cM$ specified by \FOO s-t tgds.  The main result of this subsection is
the following, which shows that whether an f-block type is
realized in the core universal solution at a given sequence of
constants $\vec{a}$ is something that can be tested by a first-order
query on the source.

\begin{proposition}\label{prop:precon}
  For each $t(\vec{x};\vec{y})\in\textsc{Types}_\cM$ there is a \FOO
  query $precon_{t}(\vec{x})$ such that for every source instance $I$
  with core universal solution $J$, and for every tuple of constants
  $\vec{a}$, the following are equivalent:
\begin{enumerate}
\item $\vec{a}\in precon_{t}(I)$
\item $t(\vec{x};\vec{y})$ is realized in $J$ at $\vec{a}$.
\end{enumerate}
\end{proposition}

\begin{proof}
  We first define an intermediate formula $precon'_{t}(\vec{x})$ that
  almost satisfies the required properties, but not quite yet.  For
  each f-block type $t(\vec{x};\vec{y})$, let
  $precon'_{t}(\vec{x})$ be the following formula:

  \[\begin{split} certain_{\cM}(\exists \vec{y}.\bigwedge t)(\vec{x}) ~~\land~~
     \bigwedge_{i\neq j} \neg certain_{\cM}(\exists
     \vec{y}_{-i}.\bigwedge t[y_i/y_j])(\vec{x}) \qquad\qquad\qquad \\ ~~\land~~
     \bigwedge_{i} \neg \exists x'.certain_{\cM}(\exists
     \vec{y}_{-i}.\bigwedge t[y_i/x'])(\vec{x},x')
  \end{split}
  \]
  where $\vec{y}_{-i}$ stands for the sequence $\vec{y}$ with $y_i$
  removed, and $t[u/v]$ is the result of replacing each occurrence of
  $u$ by $v$ in $t$. By construction, if $precon_t(\vec{a})$ holds in
  $I$, then every universal solution $J$ satisfies
  $t(\vec{a},\vec{N})$ for some some sequence of distinct nulls
  $\vec{N}$. Still, it may not be the case that $t(\vec{x};\vec{y})$
  is realized at $\vec{a}$, since it may be that that
  $t(\vec{a},\vec{N})$ is part of a bigger f-block. To make things
  more precise, we introduce the notion of an embedding.
For any two f-block types,
$t(\vec{x};\vec{y})$ and $t'(\vec{x}';\vec{y}')$, an \emph{embedding}
of the first into the second is a function $h$ mapping $\vec{x}$ into
$\vec{x}'$ and mapping $\vec{y}$ injectively into $\vec{y}'$, such
that whenever $t$ contains an atomic formula $R(\vec{z})$, then
$R(h(\vec{z}))$ belongs to of $t'$. The embedding $h$ is \emph{strict}
if $t'$ contains an atomic formula that is not of the form
$R(h(\vec{z}))$ for any $R(\vec{z})\in t$.  Intuitively, the existence
of a strict embedding means that $t'$
describes an f-block that properly contains the f-block described by
$t$.

Let $I$ be any source instance, $J$ any core universal solution,
$t(\vec{x};\vec{y})\in \textsc{Types}_\cM$, and
$\vec{a}$ a sequence of constants. 

\begin{trivlist}
\item \emph{Claim 1:} If $t$ is realized in $J$ at $\vec{a}$, then
  $\vec{a}\in precon'_{t}(I)$.

  \medskip\item \emph{Proof of claim:} Clearly, since $t$ is realized
  in $J$ at $\vec{a}$ and $J$ is a universal solution, the first
  conjunct of $precon'_{t}$ is satisfied. That the rest of the query
  is  satisfied is also easily seen: otherwise $J$ would not be a core.
  \hfill \emph{End of proof of claim.}
\end{trivlist}

\begin{trivlist}
\item \emph{Claim 2:} If $\vec{a}\in precon'_{t}(I)$, then either $t$
  is realized in $J$ at $\vec{a}$ or some f-block type
  $t'(\vec{x}';\vec{y}')\in\textsc{Types}_\cM$ is realized at a tuple of
  constants $\vec{a}'$, and there is a strict embedding
  $h:t\to t'$ such that $a_i=a'_j$ whenever $h(x_i)=x'_j$.

  \medskip\item \emph{Proof of claim:} It follows from the
  construction of $precon'_t$, and the definition of $\textsc{Types}_\cM$
  types, that the witnessing assignment for its truth must send all
  existential variables to distinct nulls, which belong to the same
  block. By Proposition~\ref{prop:plausible}, the diagram of this
  block is a specialization of an f-block type $t'\in\textsc{Types}_\cM$. It
  follows that $t$ is embedded in $t'$ and $\vec{a}$, together with
  possible some additional values in $\Cons$, realize $t'$.
  \hfill\emph{End of proof of claim.}
\end{trivlist}

We now define $precon_{t}(\vec{x})$ to be the following formula:

  \[\begin{split}&precon'_{t}(\vec{x}) ~~\land~~  %
  \hspace{-10mm}
  \mathop{\bigwedge_{\text{$t'(\vec{x}';\vec{y}')\in\textsc{Types}_\cM$}}}_\text{$h:t(\vec{x};\vec{y})\to t'(\vec{x}';\vec{y}')$ a strict embedding}
\hspace{-15mm}\neg\exists\vec{x}'.\Big(
\bigwedge_i (x_i=h(x_i)) ~~\land~~ 
precon'_{p'}(\vec{x}') \Big)
  \end{split}\]

  This formula satisfies the required conditions: $\vec{a}\in
  precon_{t}(I)$ iff $t(\vec{x};\vec{y})$ is realized in $J$ at
  $\vec{a}$. The left-to-right direction follows from Claim 2, while
  the right-to-left direction follows from Claim 1 together with the
  fact that $J$ is a core.  \qed\end{proof}

\subsection{Computing the side conditions of an f-block type}
\label{sec:side-conditions}

The issue we address in this subsection, namely that of non-rigid
f-block types, is best explained by an example.

\begin{example}
  Consider again schema mapping (e) in Figure~\ref{fig:examples}. This
  schema mapping is not laconic, because, when a source instance
  contains $Rab$ and $Rba$, for distinct values $a,b$, the canonical
  universal solutions will contain two null values $N$, each satisfying
  $SaN$ and $SbN$, corresponding to the two assignments $\{x\mapsto a,
  y\mapsto b\}$ and $\{x\mapsto b, y\mapsto a\}$. The essence of the
  problem is in the fact that the right-hand-side of the
  dependency is, in some sense, symmetric: it is a non-trivial
  renaming of itself, the renaming in question being $\{x\mapsto y, y\mapsto
  x\}$. According to the terminology that we will introduce below,
  the right-hand-side of this dependency is non-rigid. Schema mapping
  (e$'$) from Figure~\ref{fig:examples} does not suffer from this
  problem, because it contains $x\leq y$ in the antecedent, and we are
  assuming $<$ to be a linear order on the values in the source
  instance.  \hfill$\dashv$
\end{example}

In order to formalize the intuition exhibited in the above example, we
need to introduce some terminology. We say that two f-blocks, $B, B'$,
are \emph{copies of each other}, if there is a bijection 
$f$ from $\Cons$ to $\Cons$ and from $\Vars$ to $\Vars$
such that $f(a)=a$ for all $a\in\Cons$ and
$B'=\{R(f(v_1),\ldots,f(v_k))\mid R(v_1,\ldots,v_k)\in B\}$. In other
words, $B'$ can be obtained from $B$ by renaming null values. 

\begin{definition}
  An f-block type $t(\vec{x};\vec{y})$ is rigid if for any two
  sequences of constants $\vec{a},\vec{a}'$ and for any two sequences
  of distinct nulls $\vec{N}, \vec{N}'$, if $t(\vec{a};\vec{N})$ and
  $t(\vec{a}';\vec{N}')$ are copies of each other, then
  $\vec{a}=\vec{a}'$.
\end{definition}

The s-t tgd from the above example is easily seen to be non-rigid. 
Moreover, a simple variation of the argument in the above example shows:

\begin{proposition}\label{prop:rigidity-char}
  If an f-block type $t(\vec{x};\vec{y})$ is non-rigid, then the
  schema mapping specified by the \FO (in fact LAV) s-t tgd
  $\forall\vec{x}(R(\vec{x})\to \exists \vec{y}.\bigwedge
  t(\vec{x};\vec{y}))$ is not laconic.
\end{proposition}

In other words, if an f-block type is non-rigid,
one cannot simply use it as the right-hand-side of an s-t tgd without
running the risk of non-laconicity. 
Fortunately, it turns out that f-block types can be made rigid by the
addition of suitable side conditions. By a side condition
$\Phi(\vec{x})$ we will mean a Boolean combination of formulas of the
form $x_i< x_j$ or $x_i=x_j$.

\begin{definition}
  An f-block type $t(\vec{x};\vec{y})$ is rigid relative to a side
  condition $\Phi(\vec{x})$ if 
  for any two
  sequences of constants $\vec{a},\vec{a}'$ satisfying $\Phi(\vec{a})$ and $\Phi(\vec{a}')$
  and for any two sequences
  of distinct nulls $\vec{N}, \vec{N}'$, if $t(\vec{a};\vec{N})$ and
  $t(\vec{a}';\vec{N}')$ are copies of each other, then
  $\vec{a}=\vec{a}'$.
\end{definition}

\begin{definition}
  A side-condition $\Phi(\vec{x})$ is safe for an f-block type $t(\vec{x};\vec{y})$ if for every f-block
  $t(\vec{a},\vec{N})$ of type $t$ there is a f-block
  $t(\vec{a}',\vec{N}')$ of type $t$ satisfying $\Phi(\vec{a}')$ such
  that the two are copies of each other.
\end{definition}

Intuitively, safety means that the side condition is not too strong:
whenever a f-block type should be realized in a core universal
solution, there will be at least one way of arranging the variables
so that the side condition is satisfied. 
The main result of this subsection, which will be put to use in 
the next subsection, is the following:

\begin{proposition}\label{prop:rigidity}
  For every f-block type $t(\vec{x};\vec{y})$ there is a side
  condition $sidecon_t(\vec{x})$ such that
   $t(\vec{x};\vec{y})$ is rigid relative to $sidecon_t(\vec{x})$, and
   $sidecon_t(\vec{x})$ is safe for $t(\vec{x};\vec{y})$.
\end{proposition}

\begin{proof}
  We will construct a sequence of side conditions
   $\Phi_0(\vec{x}), \ldots, \Phi_n(\vec{x})$ 
   safe for $t(\vec{x};\vec{y})$, such that each
   $\Phi_{i+1}$ logically strictly implies $\Phi_{i}$,
   and such that $t(\vec{x};\vec{y})$ is rigid relative to $\Phi_n(\vec{x})$.
   Note that $n$ is necessarily bounded by a single 
       exponential function in $|\vec{x}|$.
  For $\Phi_0(\vec{x})$ we pick the tautology $\top$, which is trivially safe for $t(\vec{x};\vec{y})$.

  Suppose that $t(\vec{x};\vec{y})$ is not rigid relative to $\Phi_i(\vec{x})$, for some $i\geq 0$.
  By definition, this means that there are  
  two sequences of constants $\vec{a},\vec{a}'$ satisfying $\Phi_i(\vec{a})$ and $\Phi_i(\vec{a}')$
  and two sequences
  of distinct nulls $\vec{N}, \vec{N}'$, such that $t(\vec{a};\vec{N})$ and
  $t(\vec{a}';\vec{N}')$ are copies of each other, but
  $\vec{a}$ and $\vec{a}'$ are not the same sequence, i.e., 
  they differ in some coordinate. 
  Let $\psi(\vec{x})$ be the conjunction of all formulas of the form 
   $x_i< x_j$ or $x_i=x_j$ that are true under the assignment sending
     $\vec{x}$ to $\vec{a}$, and let $\Phi_{i+1}(\vec{x}) = \Phi_i(\vec{x})\land\neg\psi(\vec{x})$.
   It is clear that $\Phi_{i+1}$ is strictly stronger than $\Phi_i$. Moreover,
   we $\Phi_{i+1}$ is still safe for $t(\vec{x};\vec{y})$:
    consider any f-block $t(\vec{b},\vec{M})$ of type $t(\vec{x};\vec{y})$. 
    Since $\Phi_i$ is safe for $t$, we can find a f-block $t(\vec{b}',\vec{M}')$
    of type $t$ such that $\Phi_i(\vec{b'})$ the two blocks are copies of each other. 
    If $\neg\psi(\vec{b}')$ holds, then in fact $\Phi_{i+1}(\vec{b}')$ holds, and
    we are done. Otherwise,  we have that $t(\vec{b}',\vec{M}')$ is isomorphic to
    $t(\vec{a},\vec{N})$ and the preimage of $t(\vec{a}',\vec{N}')$  under this
     isomorphism will be again a copy of $t(\vec{b}',\vec{M}')$ (and therefore
       also of $t(\vec{b},\vec{M})$) that satisfies 
       $\Phi_i(\vec{b}')\land\neg\psi(\vec{b}')$, i.e., $\Phi_{i+1}(\vec{b}')$.
      \qed
\end{proof}

Incidentally, we believe the above construction of side-conditions is
not the most efficient possible, in terms of the size of the
side-condition obtained. It can probably be improved.

\subsection{Putting things together: constructing the laconic schema mapping}
\label{sec:putting}

\begin{theorem}
  For each schema mapping $\cM$ specified by \FOO s-t tgds, there is
  laconic schema mapping $\cM'$ specified by \FOO s-t tgds that is
  logically equivalent to $\cM$.
\end{theorem}

\begin{proof}
  We define $\cM'$ to consist of the following \FOO s-t tgds. For each
   $t(\vec{x};\vec{y})\in\textsc{Types}_\cM$, we take the \FOO s-t tgd
  $$\forall\vec{x}(precon_{t}(\vec{x}) ~\land~ sidecon_{t}(\vec{x})
  ~\to~ \exists\vec{y}.\bigwedge t(\vec{x};\vec{y}))$$
  In order to show that $\cM'$ is laconic and logically equivalent to
  $\cM$ (on structures where $<$ denotes a linear order), it is enough
  to show that, for every source instance $I$, the canonical universal
  solution $J$ of $I$ with respect to $\cM'$ is a core
  universal solution for $I$ with respect to $\cM$. 
  This follows from the following three facts:

\begin{enumerate}
\item Every f-block of $J$ is a copy of an f-block of the core
  universal solution of $I$. This follows from Proposition~\ref{prop:precon}.
\item Every f-block of the core universal solution of $I$ is a copy of an f-block 
  of $J$. This follows from Proposition~\ref{prop:plausible} and
  Proposition~\ref{prop:precon}, together with the safety part of 
   Proposition~\ref{prop:rigidity}. 
\item No two distinct f-blocks of $J$ are copies of each other.
  This follows from the rigidity part of 
  Proposition~\ref{prop:rigidity} together with the fact
  that $\textsc{Types}_\cM$ contains no two distinct f-block type 
  that are renamings of each other. \qed
\end{enumerate}
\end{proof}

Incidentally, if the side conditions are left out, then the resulting
schema mapping is still logically equivalent to the original mapping
$\cM$, but it may not be laconic. It will still satisfy a weak form of
laconicity: a variant of the chase defined in \cite{Fagin05:data},
which only fires dependencies whose right hand side is not yet
satisfied, will produce the core universal solution.

\

\section{Target constraints} \label{sec:target-constraints}

In this section we consider schema mappings with target constraints
and we address the question whether our main result can be extended to
this setting. The answer will be negative. However, first we need to
revisit our basic notions, as some subtle issues arise in
the case with target dependencies.

It is clear that we cannot expect to compute core universal solutions
for schema mappings with target dependencies by means of \FOO-term
interpretations. Even for the simple schema mapping defined by the s-t
tgd $Rxy\to R'xy$ and the full target tgd $R'xy\land R'yz\to R'xz$
computing the core universal solution means computing the transitive closure of $R$,
which we know cannot be done in FO logic even on finite
ordered structures. Still, we can define a notion of laconicity for
schema mappings with target dependencies. Let $\cM$ be any schema
mapping specified by a finite set of \FOO s-t tgds $\Sigma_{st}$ and a
finite set of target tgds and target egds $\Sigma_t$, and let $I$ be a source
instance. We define the \emph{canonical universal solution of $I$ with
  respect to $\cM$} as the target instance (if it exists) obtained
by taking the canonical universal solution of $I$ with respect to
$\Sigma_{st}$ and chasing it with the target dependencies $\Sigma_t$.
We assume a standard chase but will not make any assumptions on the 
chase order.
Laconicity is now defined as before: a schema mapping is laconic if
for each source instance, the canonical universal solution coincides
with the core universal solution.

Recall that, according our main result, we have (i) every schema
mapping $\cM$ specified by $\FOO$ s-t tgds is logically equivalent to
a laconic schema mapping $\cM'$ specified by $\FOO$ s-t tgds. In
particular, this implies that, (ii) for each source instance $I$, the
core universal solution for $I$ with respect to $\cM$ is the canonical
universal solution for $I$ with respect to $\cM'$. For the implication
from (i) to (ii) the requirement of logical equivalence turns out to
be stronger than needed: it is enough that $\cM$ and $\cM'$ are
\emph{CQ-equivalent}, i.e., have the same core universal solution
(possibly undefined) for each source instance \cite{Fagin08:towards}.
While CQ-equivalence and logical equivalence coincide for schema
mappings specified by \FOO s-t tgds (as follows from the closure under
target homomorphisms), the first is strictly weaker than
the second in the case with target dependencies \cite{Fagin08:towards}.

\begin{theorem}\label{thm:impossibility-tgd}
  There is a schema mapping $\cM$ specified by finitely many LAV s-t
  tgds and full target tgds, for which there is no CQ-equivalent laconic
  schema mapping $\cM'$ specified of \FOO tgds, target tgds and target
  egds.
\end{theorem}

\begin{proof} (sketch)
  Let $\cM$ be the schema mapping specified by the LAV s-t tgds

\begin{itemize}
\item
  $Rx_1x_2 \to R'x_1x_2$
\item
  $P_i x \to \exists y . Q_i y$ \qquad for $i\in\{1,2,3\}$. 
\end{itemize}
  and the full target tgds
\begin{itemize}
\item
  $R'xy \land R'yz \to R'xz$
\item
  $R'xx \land P_1 y \to P_3 y$
\item
  $R'xx \land P_2 y \to P_3 y$
\end{itemize}
For source instances $I$ in which the relations $R,P_1,P_2,P_3$ are
non-empty, the core universal solution $J$ will have the following
shape: $J(R')$ is the transitive closure of $I(R)$, and $J(Q_1),
J(Q_2), J(Q_3)$ are non-empty. Moreover, if $I(R)$ contains a cycle,
then $J(Q_1)=\{N_1\}$, $,J(Q_2)=\{N_2\}$ and $J(Q_3)=\{N_1,N_2\}$ for
distinct null values $N_1, N_2$, while if $I(R)$ is acyclic, $J(Q_1)$,
$J(Q_2)$ and $J(Q_3)$ are disjoint singleton sets of nulls.

Suppose for the sake of contradiction that there is a CQ-equivalent
laconic schema mapping $\cM'$ specified by a finite set of \FOO s-t
tgds $\Sigma_{st}$ and a finite set of target tgds and egds $\Sigma_t$.
In particular, for each source instance $I$, the canonical universal solution of $I$ with respect to $\cM'$ is the core universal solution of $I$ with respect to $\cM$.
 Let $n$ be the maximum
quantifier rank of the formulas in $\Sigma_{st}$. 

\begin{trivlist}
\item \textbf{Claim 1:} There is a source instance $I_1$ containing a
  cycle, such that the canonical universal solution $J_1$ of $I_1$ with respect
  to
  $\Sigma_{st}$ contains at least three nulls, one belonging only to
  $Q_1$, one belonging only to $Q_2$, and one belonging only to $Q_3$.

\medskip
\item The proof of Claim 1 is based on the fact that acyclicity is not
  first-order definable on finite ordered structures: take any two sources
  instances $I_1, I_2$ agreeing on all \FOO-sentences of quantifier
  rank $n$ such that $I_1$ contains a cycle and $I_2$ does not. We may
  even assume that $P_1, P_2, P_3$ are non-empty in both instances.

  Let $J_1$ and $J_2$ be the canonical universal solutions of $I_1$
  and $I_2$  with respect to $\Sigma_{st}$. 
  Then $J_2$ must contain at least three
  nulls, one belonging only to $Q_1$, one belonging only to $Q_2$ and
  one belonging only to $Q_3$. To see this, note that, first of all,
  $J_2$
  must be a homomorphic pre-image of the core universal solution of
  $I_2$ with respect to $\cM$. Secondly,
  if one of the relations $Q_i$ a non-empty in $J_2$, then the crucial
  information that $I_2(P_i)$ is non-empty is lost, in the sense that $J_2$
  would be a homomorphic pre-image of the source instance that is 
  like $I_2$ except that the relation $P_i$ is empty, which impies that,
  the result of chasing $J_2$ with $\Sigma_t$ must be homomorphically
  contained in the core universal solution of this modified source instance
  with respect to $\cM$, which is different from the core universal solution
  of $I_2$).

  This shows that $J_2$ must contain at least three
  nulls, one belonging only to $Q_1$, one belonging only to $Q_2$ and
  one belonging only to $Q_3$. 
  Each of these nulls must have been created by
  the application of a dependency from $\Sigma_{st}$. Since $I_1$ and
  $I_2$ agree on all \FOO-sentences of quantifier rank $n$, the
  left-hand-side of this dependency is also satisfied in $I_1$, and
  hence the same null is also created in the canonical universal
  solution of $I_1$.

\medskip
\item \textbf{Claim 2:} Let $J'_2$ be result of chasing $J_2$ with
  $\Sigma_t$ (assuming it exists). Then $J'_2$ cannot be the core
  universal solution of $I_1$ with respect to $\cM$. 

\medskip
\item The proof of Claim 2 is based on a monotonicity argument. More
  precisely, we use the fact that the left-hand-side of each target
  depedency is a conjunctive query, and hence is preserved under
  homomorphisms.  Let us assume for the sake of contradiction that
  $J'_1$ is the core universal solution of $I_1$ with respect to
  $\cM$, which contains exactly two null values, one in $Q_1\cap Q_3$
  and one in $Q_2\cap Q_3$. Let $N_1, N_2, N_3$ be null values
  belonging only to $J_1(P_1)$, only to $J_1(Q_2)$ and only to
  $J_1(Q_3)$, respectively.  It is easy to see that, during the chase
  with $\Sigma_t$, $N_3$ must have been identified with $N_1$ or $N_2$
  by means of a target egd $\phi$. A monotonicity argument shows that the same target egd $\phi$
  can be used to identify the two null values in the core universal
  solution $J'_1$ (note that the target dependencies cannot refer to
  the linear order on the constants). This contradicts the fact that
  $J'_1$ is the end-result of the chase with $\Sigma_t$. \qed
\end{trivlist}
\end{proof}

We expect that similar arguments can be used to find a schema mapping
$\cM$ specified by a finite set of LAV s-t tgds and target egds, such that
there is no CQ-equivalent laconic schema mapping specified by a finite
set of \FOO s-t tgds, target tgds and target egds.

\bibliographystyle{plain}
\bibliography{refs}

\appendix

\end{document}